\documentclass[11pt,a4paper]{article}
\usepackage{fullpage}
\usepackage{lmodern}
\usepackage[T1]{fontenc}
\usepackage[utf8]{inputenc}
\usepackage[tbtags]{amsmath}
\allowdisplaybreaks
\usepackage{amssymb,amsthm}
\usepackage{float}
\usepackage{hyperref}
\usepackage[svgnames]{xcolor}
\hypersetup{colorlinks={true},linkcolor={DarkBlue},citecolor=[named]{DarkGreen}}
\usepackage{tikz} 
\usepackage{pgfplots}
\usetikzlibrary{automata, positioning, calc, arrows}
\usepgfplotslibrary{patchplots, fillbetween}
\usepackage{hyphenat}
\usepackage[font=footnotesize]{caption}

\newtheorem{theorem}{Theorem}
\newtheorem{lemma}{Lemma}

\theoremstyle{definition}
\newtheorem{definition}{Definition}
\newtheorem{claim}{Claim}

\newcommand{\sset}[1]{\left\{ #1\right\}}

\begin{document}

\title{Blockchain Mining Games with Pay Forward\thanks{This work was 
generously 
supported by the ERC Advanced Grant 321171 (ALGAME) and by the ERC 
Advanced Grant 788893 (AMDROMA).}}

\author{Elias Koutsoupias\thanks{
		Department of Computer Science, University of Oxford.
		Email: \href{mailto:elias@cs.ox.ac.uk}{\nolinkurl{elias@cs.ox.ac.uk}}}
\and Philip Lazos\thanks{
		Department of Computer Science, Sapienza University of Rome.
		Email: \href{mailto:plazos@gmail.com}{\nolinkurl{plazos@gmail.com}}}
\and Paolo Serafino\thanks{
		Gran Sasso Research Institute. 
		Email: 
		\href{mailto:paolo.serafino@gssi.it}{\nolinkurl{paolo.serafino@gssi.it}}}
\and Foluso Ogunlana\thanks{
		Consensys.
		Email: 
		\href{mailto:foogunlana@yahoo.com}{\nolinkurl{foogunlana@yahoo.com}}}}

\newcommand{\pf}{w}
\newcommand{\PF}{\textsc{PF}}
\newcommand{\frontier}{\textsc{Frontier}}
\newcommand{\strictfrontier}{\textsc{StrictFrontier}}
\newcommand{\totalgain}{g^\star}
\def\bitcoin{\leavevmode\rlap{\hskip.5pt-}B}

\pgfplotsset{width=8cm}

\maketitle

\begin{abstract}
We study the strategic implications that arise from adding one
extra option to the miners participating in the bitcoin
protocol. We propose that when adding a block, miners also have
the ability to pay forward an amount to be collected by
the first miner who successfully extends their branch, giving
them the power to influence the incentives for mining. We
formulate a stochastic game for the study of such incentives and
show that with this added option, smaller miners can guarantee
that the best response of even substantially more powerful miners
is to follow the expected behavior intended by the protocol
designer.
\end{abstract}

\section{Introduction}
Bitcoin, currently the most-widely used cryptocurrency, was introduced
in a 2009 white paper~\cite{nakamoto} by the mysterious Satoshi
Nakamoto as a form of decentralized, distributed, peer-to-peer digital
currency. The backbone of the bitcoin
protocol is the \emph{blockchain}, a distributed ledger that (ideally)
takes the form of a chain where bitcoin \emph{transactions} are stored
into \emph{blocks}.  Blocks are created by special nodes of the
bitcoin network called miners that: (\emph{i}) collect and validate
the set of transactions to be included in a block and (\emph{ii})
solve a crypto-puzzle (the so-called \emph{proof of work}) that
cryptographically links the newly created (\emph{mined}) block to the
tail of the existing blockchain.  The main purpose of the blockchain
is to solve the problem of \emph{distributed consensus}, where the
consensus to be obtained is on the history (and relative order) of the
bitcoin transactions.

Since mining new blocks is a computationally intensive (and expensive)
task, miners need a reward scheme to keep mining blocks.  Bitcoin has
two main reward schemes for miners: (\emph{i}) \emph{transaction
	fees}, that are left by bitcoin users on a voluntary basis (and
hence vary in frequency and size) and \emph{coinbase} (or block)
rewards, which are constituted by a fixed set amount (currently
12.5\bitcoin{}) of newly minted bitcoins. Coinbase rewards are the
only way by which bitcoins can ever be created, and the protocol
specifies that (to avoid inflation) coinbase rewards be decreasing
over time, until the limit of 21 million bitcoins are created: at that
point the only incentive scheme left for miners will be transaction
fees. Currently, most of the rewards for mining a block are coinbase.

In a distributed setting populated by selfish miners, however, the
blockchain will hardly be a chain at all, but rather a
\emph{tree}. For instance, network delays may lead miners to add newly
mined blocks to different blocks that they believed to be the tail of the
chain. For this reason, newly mined blocks are not part of the
consensus until sufficient time has passed since their creation (i.e.,
$d$ levels of other blocks are added to the blockchain -- currently
$d=100$), and block rewards become available to miners only then.
Worse still, even though the bitcoin protocol prescribes that miners
should mine from the last known block in the chain (the so-called
\frontier{} strategy), self-interested miners may try to create a
\emph{fork} by intentionally adding a sequence of blocks (a
\emph{branch}) constituting a parallel history of the transactions, in
an attempt to reap more block rewards or to double-spend bitcoins
(once per each branch of the fork).

It should be clear from now that the bitcoin protocol is rife with
game-theoretic issues. Nakamoto \cite{nakamoto} analyzed in a
simple model, providing rough estimates showing that if a large
majority of miners follow \frontier{}, then their chain will be the
longest and contain the agreed upon history.

Since Nakamoto's paper, there has been significant work about the
strategies of miners, mainly under the assumption that the reward per
block is fixed (see for example \cite{Eyal2014,
	sapirshtein2016optimal, blockchain2016}).  In particular, these
\emph{mining games} have been systematically evaluated through the
lens of game-theory by Kiayias et al.~\cite{blockchain2016}. They
considered mining games in which miners can influence the blockchain
in two ways: by strategically choosing to mine at different branches
in an attempt to overtake the longest branch and by withholding their
mined blocks and releasing them at the right time, wasting everyone
else's computational power.  In~\cite{sapirshtein2016optimal,
	blockchain2016}, it was shown by both formal proofs and simulations
that the Nakamoto protocol is stable when no miner has computation
power more than $0.33$ of the total. Furthermore, it was also shown that when 
the
miners do not withhold blocks, the protocol is stable if no miner has
computation power more than $0.42$. 


This raises the issue of \emph{fairness} and \emph{compliance} in the
bitcoin protocol. By deviating from the protocol and mining
strategically, large miners could affect the network in two ways.  By
increasing their own rewards disproportionately to their computational
power they limit the mining rewards claimed by smaller miners, which
is unfair. Also, their constant forking in an attempt to add more
blocks upsets the protocol and can cause delays, as well as negate
transactions that have been already processed.

In this work, we propose a slight modification to the protocol whereby
miners can entice other miners to mine at their block by adding a
\emph{pay-forward} amount. This pay-forward amount is collected by the
miner of the next block, thus providing incentives to other miners to
try to extend the tree from this particular branch. When we extend the
available strategies for the miners by adding the ability to
pay-forward, we reclaim some stability by ensuring that the attacking
(computationally powerful) miner is incentivized to follow the
protocol, even if he still reaps more than his fair share of rewards.

The main technical contribution is the study of mining
games in which the strategies are extended by pay-forward. We show
that with this extension, the stability of the protocol increases
significantly: the computational power needed by a dishonest miner to
disrupt the blockchain is substantially higher. This shows an
interesting trade-off for small miners: on one hand they lose the
pay-forward amount but on the other hand, they provide the right
incentives to large miners to play $\frontier$ and thus secure that
their block will end up in the longest chain.

\subsection{Our Results}
We consider two types of \emph{stochastic games}, whose states are
rooted trees. The game is played in discrete time-steps. At the
beginning of each time-step, every miner chooses a block and tries to
mine from it. Each player $i$ has probability $p_i$ to mine a new
block, proportional to his computational power. In the end, the new
block may be added to the blockchain or kept hidden and released
later. After many rounds of playing, the utility of each player is the
fraction of bitcoins he owns (from block and pay-forward rewards) over
the total value that has been mined in the longest chain.

Let $p$ be the computational power of the strongest player, called
Miner 1. If $p$ is large enough, his best response, given that
everyone else is playing $\frontier$, may be a different strategy. We
find thresholds on $p$ and $w$, the pay forward amount of every other
player, to guarantee that mining at the end of the longest chain is a
best response for Miner 1. As in \cite{blockchain2016}, we consider
two variants. In the immediate release case, any mined block has to be
added to the blockchain immediately for other players to use. In this
case, without pay forward rewards, it was proven in
\cite{blockchain2016} that $\frontier$ is a best response for Miner 1
for $p < h$, where $0.361 \le h \le 0.455$. Experimentally, they also
showed that there is always a deviating strategy for $h \approx
0.42$. In the strategic release case, the newly mined blocks are
public knowledge, but can only be used by other players when their
creators decide to \emph{release them}. As before, it was shown in
\cite{blockchain2016} that there is a threshold $\hat{h} \ge 0.308$
(experimentally shown to be approximately $0.33$ \cite{blockchain2016,
	sapirshtein2016optimal}) such that $p < \hat{h}$ leads to
$\frontier$ as a best response for Miner 1.

We improve these thresholds to $h = 0.5$ and $\hat{h} \ge 0.344$ by
showing that there exists some $w$ (as a function of $p$) that
guarantees Miner 1's best response is to play $\frontier$ when all
remaining miners play $\frontier$ and pay forward $w$. In addition, we
experimentally show that for $h \le 0.44$ that outcome is also a pure
Nash equilibrium for the immediate release case, while for the
strategic release case we find that $\hat{h} \approx 0.38$. We also
devise a linear program to calculate
the minimum value of $w$ for any given $p$.


\subsection{Related Work}\label{sec:related_work}
Bitcoin, originally introduced by Nakamoto in \cite{nakamoto} and
followed by several cryptocurrencies, such as Litecoin, Ethereum and
Ripple, initiated the research on blockchains.  Nakamoto's original
paper analyses double spending attacks while Rosenfeld provides a more
detailed analysis \cite{rosenfeld2014analysis}.  Bonneau et
al. \cite{bonneauetal} and Tschorsch and Scheuermann \cite{tschorsch}
provided extensive surveys of the research and challenges in
cryptocurrencies.

Kroll et al. \cite{kroll} was one the first papers to consider the
economics of Bitcoin mining, assuming that participants behave
according to their incentives. Eyal and Sirer \cite{Eyal2018} showed
that the security of bitcoin is not guaranteed by a majority of honest
miners as was previously assumed. They gave a specific mining strategy
and argued that a pool of miners, with at least a $1/3$ of the total
processing power, can get extra profit regardless of the block
propagation characteristics of the network.  With sufficiently
favorable block propagation, this threshold of processing power falls
to $0$. Extending this, Sapirshtein et
al. \cite{sapirshtein2016optimal} provided systematic analysis of the
space of selfish mining strategies based on computational results.

A similar approach was taken by Kiayias et al. \cite{blockchain2016}
who provided a framework for studying the strategic considerations
made by Bitcoin miners. They formulated two abstract stochastic games
and proved rigorous bounds on the threshold of computational power
below which the honest strategy is a Nash Equilibrium. Carlsten et
al. \cite{carlsten2016instability} extended the previous approaches by
investigating mining games when the reward for miners varies and comes
mainly from transaction fees. They observed that the random block
arrival times lead to high variance in rewards and they considered
strategies that lead to instability. They showed via simulation that
an equilibrium exists, but it has the counter intuitive and undesirable
effect of a growing backlog of unprocessed transactions.

Several other studies examine possible attacks on the protocol and
suggest adaptations to ensure its security. Rosenfeld
\cite{rosenfeld_pool} and Courtois and Bahack
\cite{DBLP:journals/corr/CourtoisB14} discuss pool mining
attacks. Eyal \cite{Eyal2014} introduces an attack in which mining
pools infiltrate one another resulting in a pool game. Lewenberg et
al.~\cite{Lewenberg2015} also provide a game theoretic analysis of
pool mining. Babaioff et al.~\cite{Babaioff2011} consider Sybil
attacks on the network and propose a reward scheme to prevent miners
from hiding transactions in competition with other miners.

There is also considerable work on the performance and scalability of
blockchain inspired networks. Sompolinsky and Zohar's
\cite{sompolinsky2015secure} Greedy-Heaviest-Observed-Sub-Tree rule is
an alternative consensus mechanism. Eyal \cite{eyal2016bitcoin}
proposes Bitcoin-NG which increases the throughput of Bitcoin. Poon
and Dryja's Lightning Network \cite{poon2016bitcoin} scales via
off-chain transactions and hashed commitments. Sompolinsky and Zohar's
PHANTOM \cite{DBLP:journals/iacr/SompolinskyZ18} achieves near
unlimited transaction throughput. SPECTRE
\cite{DBLP:journals/iacr/SompolinskyLZ16}, by Sompolinsky et al, fully
orders the transactions in blocks using recursive elections and can
used in combination with PHANTOM. Kiayias et al. developed
Ouroboros\cite{ouroboros}, a proof of stake network with provable
security, utilising a secure multi-party coin flipping protocol.

\section{Model and Notation}
The model contains subtleties in the abstractions that at first glance
might be overlooked by readers who are too familiar with the bitcoin
protocol. We do our best to point out any of those instances, but be
especially vigilant when reading about the way payoffs are calculated.\\

\noindent
The bitcoin mining game with pay forward is an abstraction of the
actual protocol, simplified in a way that can only accentuate its game
theoretic issues (i.e. negative results hold a fortiori for the actual
protocol) while being open to rigorous analysis. The parameters of the
game are:
\begin{enumerate}
	\item the number $n$ of miners (or players).
	\item the probabilities $p_1,\ldots, p_n$, representing the \emph{hash
		power} of each miner, which is the probability that they solve the
	crypto-puzzle. They are proportional to the computational power of
	each miner and such that: $p_i\geq 0$ for each $i$ and
	$\sum_{i=1}^n p_i = 1$.
	\item the depth $d$ after which mined blocks are `paid', i.e. their
	coinbase becomes usable by the respective miner. Without loss of
	generality we mostly consider $d=\infty$ which gives the attacker
	extra power.
\end{enumerate}

We assume that each block has a \emph{fixed reward}, which by
appropriate scaling is equal to 1. This is reasonable, as the large
majority of rewards come from the fixed coinbase rewards. The
fluctuating transaction fees (while not insignificant) are still
comparatively small.

During the execution of the protocol, miners add their blocks to the
blockchain one at a time\footnote{ Due to the decentralized nature of
	bitcoin, it is possible for more than one miners to mine a block
	simultaneously. This event does occur by chance, but it is very rare
	and it is not something a miner can plan for. For this reason, we
	assume that exactly one block is mined at every step, but most of
	our analysis could be easily extended to the case in which each
	miner succeeds with probability $p_i$ independently.}.  Their goal
is to maximize the fraction of their blocks on the longest branch of
the blockchain, as the longest branch represents the consensus
 and all other branches are pruned and the computational power 
spent for
their creation is wasted.  This does not mean that all the miners try
to extend the longest branch at every time step, because some times
they might benefit by either adding a block to a shorter branch or
withholding it for later.

\begin{definition}
	A public state is a rooted tree. Every node is labelled by one of the
	players and the amount of money he pays-forward. The nodes represent
	mined blocks and the label indicates the player who mined the
	block. The pay-forward amount is collected by the next miner in the
	longest branch. Every level of the tree has at most one node labelled
	$i$ because there is no reason for a player to mine twice the same
	level.
	
	A private state of player $i$ is similar to the public state except
	it may contain more nodes called private nodes and labelled by
	$i$. The public tree is a subtree of the private tree and has the
	same root.
\end{definition}

The incomplete information case (where the $p_i$'s or private states
are unknown) is significantly more complicated. In this work we only
consider the full information case where even the private states are
common knowledge (but only become part of the public state when their
respective miners add them). We consider two variants:
\begin{description}
	\item[Immediate-release model] Every mined block is added to the
	public tree immediately.
	\item[Strategic-release model] Mined blocks can be withheld: every
	miner is aware of their existence, but they can only mine from them
	once they are added to the public tree.
\end{description}
The second model has no counterpart in practice, but it serves as an
intermediate step between the full and incomplete information models
and allows us to study issues around strategic release of blocks. If a
strategy is not dominant in this model it cannot be dominant in the
incomplete information setting.

\begin{definition}[Strategy]
	The strategy of miner $i$ can be fully characterized by three
	functions $\mu_i,r_i, \PF_i$:
	\begin{itemize}
		\item the mining function $\mu_i$ selects a block of the public
		state to mine from
		\item the release function $r_i$ which is the rooted tree he
		releases to the public. It is a subtree of his private state that
		contains the public state.
		\item the pay-forward function $\PF_i$ which is the amount of money
		to be left for the next miner.
	\end{itemize}
	Both functions depend on the public and private states of every player.
\end{definition}

The original suggested strategy in \cite{nakamoto} is $\frontier$.
\begin{definition}[\frontier]
	A miner follows strategy $\frontier$ if he releases mined blocks
	immediately, always mines at the deepest node in the blockchain and
	pays forward 0 \bitcoin.
\end{definition}
In the following, we refer to $\frontier(w)$ as the strategy that mines and 
releases like $\frontier$ but always pays forward $w$.

The game is played in \emph{phases}. At each phase exactly one miner
will mine a block. Then he will choose if he wants to release it, which
may trigger a cascade of releases from other players. Eventually, once
no one else wants to release anything the public knowledge is updated
and the next phase begins.  Even if no one releases a block, miners
know when the phase ends since everything is public knowledge.

To incentivise miners to mine new blocks, payments are necessary. When
a block is added, it is unclear if it is going to remain on the
longest branch permanently.  To remedy this, blocks are paid for only
after their branch is increased by $d$ blocks, after which time it is
safe to assume it will stay in the longest branch.  Blocks are paid
through coinbase, transaction fees and pay-forward.  Currently,
coinbase rewards (which can be claimed at $d=100$) far outweigh
transaction fees.  However, by design Bitcoin will eventually do away
with coinbase rewards to limit inflation, at which point only
transaction fees will remain, leading to potential
instability~\cite{carlsten2016instability}. We only consider coinbase
rewards in this work. For game theoretic analysis we rigorously define
payments:
\begin{definition}[Payments]
	For some nodes of the tree, the miners who discovered them will get
	a payment (coinbase rewards are normalized to 1). The payments
	comply with the following rules:
	\begin{itemize}
		\item the blocks that receive a payment must form a path from
		the root.  This immediately adds the restriction that at
		every level of the tree exactly one node receives payment.
		\item among the blocks of a single level that satisfy the
		above condition, the first one which succeeds in having a
		descendant $d$ generations later receives payment.
	\end{itemize}
	Since only one block per level is paid for and they form the longest
	branch, the utility of a miner in the long run is defined as the
	fraction of pay-forward and coinbase rewards he obtained minus the
	pay-forward he paid in that branch, over it's length.
	
\end{definition}
When a block is paid for, strategic miners will ignore any branch that
starts at an earlier node. This limits the shape of the public state
to a long path (the \emph{trunk}) with ignored \emph{stale} branches
dangling from it. Only at the end of the trunk we may find multiple
active branches, competing to reach length $d$ and render the others
stale. We will also consider the case where $d=\infty$. Here there is
the possibility that two competing branches will go on forever, but
since one of them will have less computational power, a case of
gamblers ruin makes this impossible.

\section{Immediate Release}\label{sec:immediate_release}
In this section we show that if $\frontier(w)$  (for some value of $w$) is followed 
by all miners but one, the remaining miner's best
response is $\frontier$, provided his hash power is $p < 0.5$. Since
only one miner might deviate, we can view this situation as Miner 1
being the potentially deviant large miner with hash power $p<0.5$ and
Miner 2 all small miners combined into one, which we call the honest
miners, since they follow the same strategy.

In essence, every honest player mines at the deepest node and always
pays forward $w$. This amount propagates across the small miners
forming Miner 2 but Miner 1 can claim it for himself, leading to an
interesting stochastic game where Miner 1 is incentivise to mine the
longest chain, since the rewards don't accumulate.

The blockchain itself is a rooted tree, but after pruning abandoned
branches we are left with one long, undisputed, main path followed by
two branches of length $a$ and $b$, mined by Miner 1 and Miner 2
respectively. Therefore, the possible states of the game have the form
$(a,b,c)$, where $a \le b + 1$ since Miner 2 always mines the longest
branch (except briefly, right when Miner 1's branch overtakes the
longest) and $c$ is either 0 or 1, to indicate whether the \emph{last
	block before the fork} contained a pay forward reward.

\begin{figure}[H]
	\begin{center}
		\begin{tikzpicture}[scale=.7, transform shape]
		\tikzstyle{every node} = [rectangle, fill=gray!30,font=\Large]
		\node (x1) at (-3, 0) {1};
		\node (x2) at (-2, 0) {2};
		\node (x3) at (0, 0) {2};
		\node (x4) at (1, 0) {1};
		\node (x5) at (2, 0) {2};
		\node (x6) at (3, 1) {2};    
		\node (x7) at (3, -1) {1};        
		\node (x8) at (4, -1) {1};    
		\node (x9) at (5, -1) {1};    
		\node (x10) at (4, 1) {2};    
		\node (x11) at (5, 1) {2};    
		\node (x12) at (6, 1) {2};    
		\node (x13) at (7, 1) {2};    
		
		\draw[black,<-] (x1) -- (x2);
		\draw[black,<-,dotted] (x2) -- (x3);
		\draw[black,<-] (x3) -- (x4);
		\draw[black,<-] (x4) -- (x5);
		\draw[black,<-] (x5) -- (x6);
		\draw[black,<-] (x5) -- (x7);
		
		\draw[black,<-] (x7) -- (x8);
		\draw[black,<-] (x8) -- (x9);
		\draw[black,<-] (x6) -- (x10);
		\draw[black,<-] (x10) -- (x11);
		\draw[black,<-] (x11) -- (x12);
		\draw[black,<-] (x12) -- (x13);
		
		\draw[black,|-|] (-3.2,-1)
		--node[draw=none,fill=none,below]{Trunk} (2,-1); 
		
		\draw[black,|-|] (2.8,-2)
		--node[draw=none,fill=none,below]{$a=3$} (5.2,-2); 
		
		\draw[black,|-|] (2.8,2) --node[draw=none,fill=none,above]{$b=5$} (7.2,2);
		\end{tikzpicture}
	\end{center}
	\caption[A typical state tree.]{A typical state tree. The trunk
		represents the blocks whose rewards have already been
		collected. The current state $(3,5,1)$ of the game is represented
		by the blocks mined by Miner 1 and Miner 2 and $c=1$ because Miner
		2 controls the block before the fork.}
\end{figure}
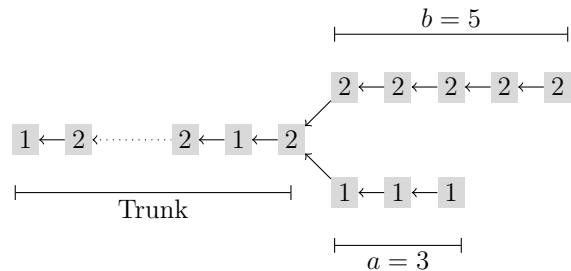

The set of states $(a,b,c)$ can be partitioned as follows:
\begin{enumerate}
	\item \textbf{Winning states}: the set $W$ of states where Miner 2
	capitulates and starts mining from the tip of Miner 1's branch. This
	happens exactly when $a = b + 1$, therefore
	$W = \{(b+1,b,c) \mid b \ge 0 \text{ and } c \in
	\sset{0,1}\}$. After Miner 1 overtakes, the new state of the game is
	$(0,0,0)$ since they both mine at the same point.
	
	\item \textbf{Capitulation states}: the set $C$ of states where Miner
	1 capitulates, abandons his branch and mines from a block of Miner
	2's branch, thus moving the game to state $(0,s,1)$ for some $s$
	s.t. $0\leq s<b$.  We say that Miner 1 capitulates \emph{at}
	$(a,b,c)$ and \emph{to} $(0,s,1)$. Clearly, after capitulating there
	would be a pay forward reward available.
	
	\item \textbf{Mining states}: the set $M$ of states where Miner 1 and
	Miner 2 mine their respective branches.
\end{enumerate}
Miner 1 can capitulate to any state $(0,s,1)$ and will always choose
the one that maximizes his payoff. Since he is rational, when
capitulating from state $(a,b,c)$ he would only go to states
$(0, s, 1)$ with $s < b$, otherwise he would be undercutting his own
tentative branch. The strategy $\frontier$ has
$M = \sset{(0,0,0), (0,0,1)}$, $C =
\sset{(0,1,0),(0,1,1)}$ and always capitulates to $(0,0,1)$.\\

Let $g_k(a,b,c)$ denote the optimal expected gain of Miner 1 starting
from state $(a,b,c)$ when the longest chain is extended by $k$
levels. Knowing Miner 1 will never pay-forward, we can recursively
define:
\begin{equation}\label{eq:gain_at_k_immediate}
	g_k(a,b,c) = 
	\begin{cases}
		g_{k-1}(0,0,0) + a + w\cdot c \quad\quad\text{if } a = b + 1\\
		\max 
		\begin{cases}
			\max_{s=0,\ldots,b-1}g_k(0,s,1)\\
			pg_k(a+1,b,c) + (1-p)g_{k-1}(a,b+1,c) 
		\end{cases}
	\end{cases}.
\end{equation}
The first line represents Miner 2 capitulating to Miner 1's chain. The 
second line containing the $\max$ has two cases. In the first case, Miner 1 
capitulates to the optimal point within Miner 2's chain. In the second, they both 
mine and extend their chains with probability $p$ and $(1-p)$ respectively. 
Notice that the first term is $pg_k(a+1,b,c)$ and not $pg_{k-1}(a+1,b,c)$. Since 
$a < b + 1$, if Miner 1 adds a block the deepest level of the blockchain does 
not change, thus $k$ remains the same.

Clearly, Miner 1 eventually will add some block to the chain and the
game will restart at state $(0,0,0)$. Therefore, we expect that the
initial state $(a,b,c)$ has no effect asymptotically. We define the
expected gain per level $\totalgain$ as:
\begin{equation}\label{def:expected_gain}
	\totalgain = \lim_{k\rightarrow \infty} \frac{g_k(a,b,c)}{k}.
\end{equation}
We can decompose $\totalgain$ into two separate quantities:
\begin{equation}\label{eq:decompose_gain}
	\totalgain = q_M + q_{\PF}\cdot \pf,
\end{equation}
where $q_M$ is the expected fraction of blocks he mined and $q_\PF$
the expected fraction of blocks mined and containing pay forward
rewards.
\begin{lemma}\label{lemma:honest_miner_gain}
	If Miner 1 follows $\frontier$, we have $q_M = p$ and
	$q_\PF = p(1-p)$.
\end{lemma}
\begin{proof}
	If Miner 1 is honest then the probability that he mined any
	arbitrary block is $p$. That block also contains pay forward rewards
	if the one that immediately preceded it was mined by Miner 2, which
	occurs with probability $p(1-p)$.
\end{proof}

We are ready to state the main theorem of this section.
\begin{theorem}
	For every $p < 0.5$, there exists $w \ge 0$ large enough so that if
	every miner but one follows $\frontier(w)$, the best response of the
	remaining miner with hash power $p$ is $\frontier$.
\end{theorem}
\begin{proof}
	We will show that $q_\PF$ attained by $\frontier$ is at least as
	large as that of any optimal strategy, and in particular it is
	larger for $0 < p < 0.5$.
	\begin{lemma}\label{thm:frontier_more_payforward}
		If $\frontier$ is not the optimal strategy and $0<p<0.5$, we have
		that $p(1-p) > q_\PF$. For $p=0$ or $p=0.5$ we have equality.
	\end{lemma}
	\begin{proof}
		Assuming $\frontier$ is not an optimal strategy, we consider
		different cases, depending on the actions of the optimal strategy
		starting from state $(0,1,c)$.
		
		Clearly, we have that $g_k(0,1,1) \ge g_k(0,1,0)$ for all $k$. To
		see this, consider what happens if we apply exactly the same
		sequence of actions starting from $(0,1,0)$ and from $(0,1,1)$. As
		the only difference between the two states is the initial pay
		forward amount, if we start from $(0,1,1)$ the reward will be the
		higher (in expectation) by some fraction of $w$.  Therefore, if
		the optimal strategy capitulates at $(0,1,1)$ it must also
		capitulate at $(0,1,0)$. This is exemplified in Figure
		\ref{fig:frontier}, where blocks with a red outline represent
		those mined by the optimal strategy, and red arrows show how the
		action of the optimal strategy in one case implies the same action
		in the other.
		
		\begin{figure}[H]
			\centering
			\begin{tikzpicture}[scale=.65, transform shape]
			\tikzstyle{every node} = [rectangle, fill=gray!30, font=\huge]
			\node (x1) at (-3, 0) {1};
			\node (x2) at (-2, 0) {2};
			\node (a) at (0, 0) {1};
			\node[draw=red, thick] (b) at (2, 1.5) {2};
			
			\node (y1) at (4, 0) {1};
			\node (y2) at (5, 0) {2};
			\node (c) at (7, 0) {2};
			\node[draw=red, thick] (d) at (9, 1.5) {2};
			
			\draw[black,->] (x2) -- (x1);
			\draw[black,->,dashed] (a) -- (x2);
			\draw[black,->] (y2) -- (y1);
			\draw[black,<-,dashed] (y2) -- (c);
			\draw[black, dotted, ultra thick] (3, -1) -- (3, 1);	
			\draw[black,<-] (a) -- (b);
			\draw[black,<-] (c) -- (d);
			\draw[red,->, thick] (d)  .. controls($(d)!0.25!(b) + (0,1.5)$) and 
			($(d)!0.75!(b) + 
			(0,1.5)$) ..  (b);
			\end{tikzpicture}
			\caption{Mining like Frontier}
			\label{fig:frontier}
		\end{figure}
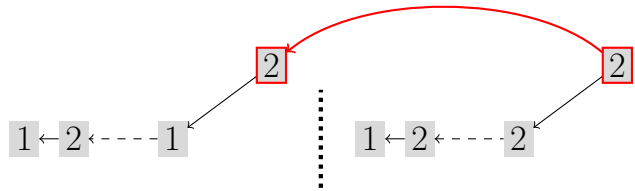
		
		But in this case, the optimal strategy capitulates at exactly the
		same states as $\frontier$, as shown in
		Figure~\ref{fig:frontier}. Since the blockchain starts empty, at
		$(0,0,0)$ this strategy would behave exactly the same as $\frontier$ and
		would therefore not be optimal.
		
		\begin{figure}[H]
			\centering
			\begin{tikzpicture}[scale=.65, transform shape]
			\tikzstyle{every node} = [rectangle,
			font=\huge, fill=gray!30]
			\node (x1) at (-3, 0) {1};
			\node (x2) at (-2,0) {2};
			\node[draw=red, thick] (a) at (0, 0) {1};
			\node (b) at (2, 1.5) {2};
			
			\node (y1) at (4, 0) {1};
			\node (y2) at (5, 0) {2};
			\node[draw=red, thick] (c) at (7, 0) {2};
			\node (d) at (9, 1.5) {2};
			
			\draw[black,<-] (x1) -- (x2);
			\draw[black,<-,dashed] (x2) -- (a);
			\draw[black,<-] (y1) -- (y2);
			\draw[black,<-,dashed] (y2) -- (c);
			\draw[black, dotted, ultra thick] (3,
			-1) -- (3, 1); \draw[black,<-] (a) --
			(b); \draw[black,<-] (c) -- (d);
			\draw[red,->, thick] (a)
			.. controls($(a)!0.25!(c) + (0,-1.5)$)
			and ($(a)!0.75!(c) + (0,-1.5)$) ..
			(c);
			\end{tikzpicture}
			\caption{Case 1}
			\label{fig:case1}
		\end{figure}	
		In Case 1 (shown in Figure~\ref{fig:case1}), we assume that the
		optimal strategy mines at $(0,1,0)$. As before, since $(0,1,1)$ is
		a more beneficial state it would mine from there as well. We can
		therefore assume that the optimal strategy always capitulates to
		states $(0,s,1)$ with $s \ge 1$. Blocks are permanently added to
		the blockchain only after a capitulation.  Given that, in order to
		compute $q_{PF}$ for Miner 1, we construct the Markov chain shown
		in Figure~\ref{fig:markov1}, with states $(0,0,0)$ and $(0,s,1)$
		indicating that Miner 2 or Miner 1 capitulated respectively. Each
		transition is labelled with a pair where the first element is the
		transition probability, whereas the second element is the minimum
		number of blocks added to the chain when the transition occurs.
		\begin{figure}[H]
			\centering
			\begin{tikzpicture}
			\node[state] (win) {(0,0,0)};
			\node[state, right=2.5 of win] (lose) {($0,s,1$)};
			
			\draw[every loop] (win) edge[bend right, auto=right]
			node[below] {$(1-p)\alpha, 1$} (lose) (lose) edge[bend right,
			auto=right] node[above] {$\beta, 2$} (win)
			
			(win) edge[in=200, out=230, loop] node[below]
			{$(1-p)(1-\alpha),2$} (win) (win) edge[in=120, out= 150, loop]
			node[above] {$p,1$} (win)
			
			(lose) edge[loop right] node[right] {$1-\beta,1$} (lose);
			\end{tikzpicture}
			\caption{Case 1 Markov chain}
			\label{fig:markov1}
		\end{figure}
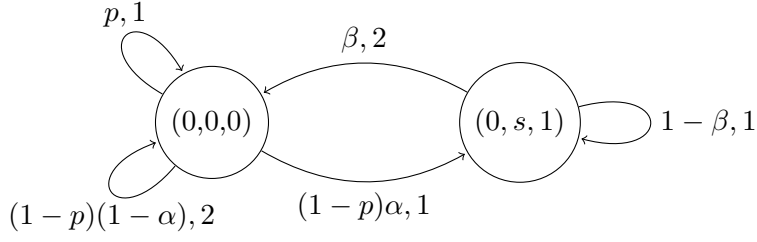
		
		Starting from $(0,0,0)$, with probability $p$ Miner 1 adds a block
		and the game restarts. With probability $1-p$ we move to state
		$(0,1,0)$ and the optimal strategy has probability $\alpha$ of
		losing the race and capitulating to some state $(0,s,1)$. From
		there, the optimal strategy has probability $\beta$ of winning and
		returning to $(0,0,0)$.
		
		As $q_\PF$ is an asymptotic quantity, we are interested in the
		fraction of pay forward rewards per block mined at the stationary
		distribution.  \newcommand{\stat}{\pi} Let $\stat$ be the
		stationary probability of state $(0,s,1)$. Then:
		\begin{equation}
			\stat = \frac{\alpha - \alpha\cdot p}{\alpha - \alpha\cdot p + \beta}
		\end{equation}
		In the actual game, multiple blocks are potentially mined in each
		edge transition, and only the first block might claim a pay
		forward reward. On every transition at least one block is
		permanently added, \emph{except from $(0,s,1)$ to $(0,0,0)$ and
			$(0,0,0)$ to itself through transition $(1-p)(1-\alpha)$}. In
		these cases at least two blocks are added, since Miner 1 must mine
		at least two blocks to be ahead of Miner 2. From $(0,s,1)$ to
		$(0,0,0)$ is the only transition where Miner 1 wins a block
		containing pay forward. Dividing the blocks with pay-forward over
		the total amount mined at the stationary distribution:
		\begin{equation}
			q_\PF \le 
			\frac{\stat \beta}
			{(1-\stat) (2(1-\alpha) (1-p)+\alpha (1-p)+p)+\stat(2 \beta+1-\beta)}
		\end{equation}
		We now need to set the parameters $\alpha$ and $\beta$ optimally
		to maximize the upper bound on $q_\PF$. We take the derivative of
		the upper bound of $q_\PF$ with respect to $\alpha$ and $\beta$ to
		get:
		\begin{equation}
			\frac{dq_\PF}{d\alpha} =
			\frac{\beta^2 (2-p) (1-p)}{(\alpha (p-1)+\beta (p-2))^2}\geq 0
		\end{equation}
		and
		\begin{equation}
			\frac{dq_\PF}{d\beta} = 
			\frac{\alpha^2 (1-p)^2}{(\alpha (p-1)+\beta (p-2))^2}\geq 0.
		\end{equation}
		From \cite[Lemma~1]{blockchain2016}, we know that the probability
		that Miner 1 reaches a winning state starting from $(a,b)$ is at
		most $(p/(1-p))^{(b-a+1)}$. Therefore, the optimal strategy must
		have
		\begin{equation}
			\alpha \le 1 \text{ and } \beta \le \left(\frac{p}{1-p}\right)^{s+1} \le 
			\left(\frac{p}{1-p}\right)^2
		\end{equation}
		Since $q_\PF$ is an increasing function of $\alpha$, $\beta$, the
		upper bound is obtained by setting them to their highest possible
		value, leading to:
		\begin{equation}
			q_\PF \le
			\frac{(1-p) p^2}{1-(1-p) p (3-2p)}.
		\end{equation}
		Compared to the honest outcome $p(1-p)$, we have:
		\begin{equation}
			p(1-p) - \frac{(1-p) p^2}{1-(1-p) p (3-2p)}
			=
			\frac{(1-p) p (1-2 p)}{1-(1-p) p (3-2 p)} \ge 0,
		\end{equation}
		with equality obtained only for $p=0$ and $p=0.5$. We also plot
		this result at the end of the proof.
		
		In the last case, the optimal strategy capitulates at $(0,1,0)$
		but mines at $(0,1,1)$.
		
		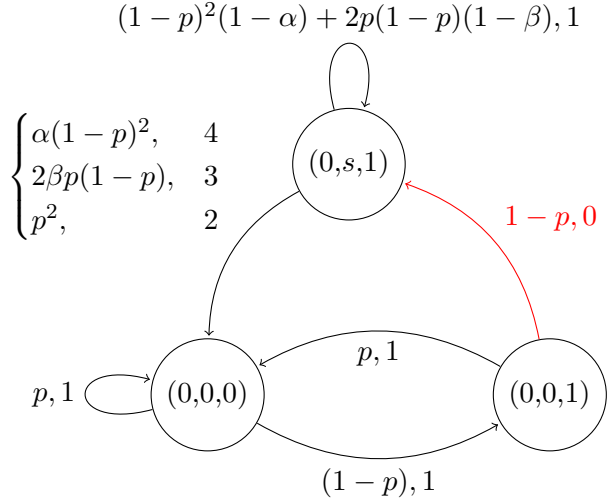
\begin{figure}[H]
			\centering
			\begin{tikzpicture}[scale=1, transform shape]
			\node[state]             (win) {(0,0,0)};
			\node[state, right=3 of win] (inter){(0,0,1)};
			\node[state, above right=2 and 0.8 of win] (lose) {(0,$s$,1)};
			
			\draw[every loop]
			(win) edge[bend right, auto=right]  node {$(1-p),1$} (inter)
			(win) edge[loop left]             node {$p,1$} (win)
			
			(inter) edge[bend right, auto=left] node {$p,1$} (win)
			(inter) edge[bend right, auto=right,red]  node {$1-p,0$} (lose)
			
			(lose) edge[bend right, auto=right] node 
			{
				$\begin{cases}
				\alpha (1-p)^2, & 4\\
				2 \beta p (1-p), & 3\\
				p^2, & 2
				\end{cases}$
			} (win)
			
			(lose) edge[loop above]             
			node {$(1 - p)^2 (1 - \alpha) + 2 p (1 - p) (1 - \beta),1$} (lose);
			\end{tikzpicture}
			\caption{Case 2 Markov chain}
			\label{fig:markov2}
		\end{figure}
		
		The Markov chain in this case (Figure~\ref{fig:markov2}) is
		trickier.  As before each state represents a capitulation:
		$(0,0,0)$ for Miner 2 and $(0,0,1)$, $(0,s,1)$ with $s\ge 1$ for
		Miner 1. However, the transition indicated with the red arrow is
		\emph{not} a capitulation: it is merely a move to a state that
		Miner 1 could also capitulate \emph{to}. To more accurately model
		the transitions of state $(0,s,1)$, we consider multiple moves
		ahead. Specifically, three outcomes could happen before the
		optimal strategy considers capitulating. With probability $p^2$,
		Miner 1 could add two blocks and move to $(0,0,0)$.  Miner 2 could
		also add two blocks and move to $(0,s+2,1)$, with probability
		$(1-p)^2$.  Then, Miner 1 wins with probability $\alpha$. Finally,
		Miner 1 and Miner 2 could both add one block leading to
		$(1,s+1,1)$ with probability $2p(1-p)$. From there, Miner 1 wins
		with probability $\beta$.
		
		Carefully adjusting for the more complicated transitions, the same
		analysis as before holds. For both cases, we plot the difference
		$p(1-p) - q_\PF$ in Figure~\ref{fig:plots}.
		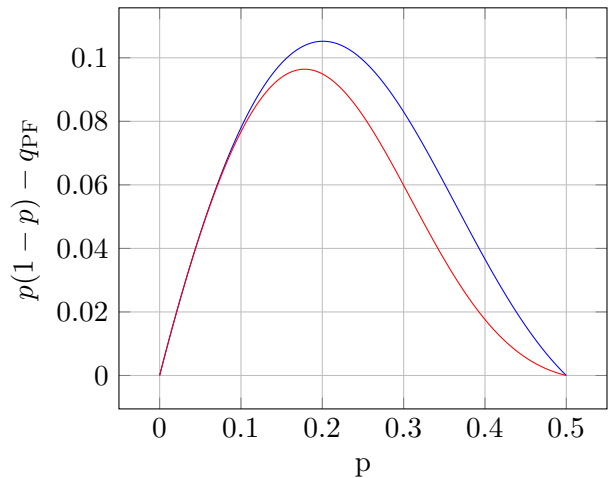
\begin{figure}[H]
			\centering
			\begin{tikzpicture}
			\begin{axis}[
			yticklabel style={/pgf/number format/fixed},
			xlabel = {p},
			ylabel = {$p(1-p)-q_\PF$},
			grid=major,
			]
			\addplot [blue, domain=0.000001:0.5, samples=201]
			{-(((-1 + x)^3 *x* (-1 + 2 *x))/(-1 + (-1 + x) *x* (-3 + 2* x)))};
			\addplot [red, domain=0:0.5, samples=201]
			{
				-(((-1 + x)^3*x*(-1 + 2*x)*(-1 + x + x^2))/(1 + (-1 +
				x)*x*(4 + x*(-4 + x*(-1  
				+ 2*x)))))
			};
			\end{axis}
			\end{tikzpicture}
			\caption[Advantage of $\frontier$ in collecting $q_\PF$ for Case
			1 and Case 2]{Advantage of $\frontier$ in collecting $q_\PF$ for
				Case 1 (blue) and Case 2 (red) }
			\label{fig:plots}
		\end{figure}
	\end{proof}
	Therefore, by (\ref{eq:decompose_gain}) and
	Lemma~\ref{lemma:honest_miner_gain} we have that
	$$
	\totalgain = q_M + q_{\PF}\cdot \pf \le p + (1-p)\cdot w \Rightarrow w \ge 
	\frac{q_M - p}{(1-p) - q_{PF}}
	$$
	and since $(1-p) > q_{PF}$ for $0 < p < 0.5$ there must be some
	value of $w$ large enough to make $\frontier$ the best response.
\end{proof}
Unfortunately, this result is not enough to calculate $w$. Giving crude upper 
bounds on $q_{M}$ or $q_{PF}$ independently is not too hard, we can use  
\cite[Lemma~1]{blockchain2016} from example. However,  because the trade-off  
between the two is hard to establish and the actual blockchain does not use $d 
= \infty$, it is more useful to develop an algorithm that finds the minimum 
$w$ required for a specific $d$, to any degree of accuracy.

\subsection{Calculating the optimal $w$ for finite $d$} \label{sec:LP}
In reality, the values of $w$ needed are not too large. We could
attempt to find the minimum $w$ necessary by computing $g_k(0,0,0)/k$
for large values of $k$ and comparing with $p+p(1-p)w$, but a naive
implementation would take $O(k^3)$ time. To simplify this search we
can define a potential as in \cite{blockchain2016}
\begin{equation}
	\phi(a,b,c) = \lim_{k \rightarrow \infty} g_k(a,b,c) - k\cdot \totalgain, 
\end{equation}
that captures the advantage of Miner 1 at different states. This
quantity is bounded, because irrespective of the initial state
$(a,b,c)$ the game repeats: in particular, when Miner 1 gets ahead of
Miner 2 they both restart at state $(0,0,0)$. Therefore, the gain per
$k$ asymptotically converges to $g^\star$ and $\phi$ is well
defined. Following recurrence \eqref{eq:gain_at_k_immediate}, we have:

\begin{equation}\label{eq:immediate_potential}
	\phi(a,b,c) = 
	\begin{cases}
		\phi(0,0,0) + a + c\cdot w - \totalgain 
		\quad\quad \text{if } a = b + 1\\
		\max
		\begin{cases}
			\max_{s=0,\ldots,b-1}\phi(0,s,1) \\
			p\cdot\phi(a+1,b,c)\\
			\quad + (1-p)\cdot (\phi(a,b+1,c) - \totalgain)
		\end{cases}
	\end{cases},
\end{equation}
where we set $\phi(0,0,0)=0$. Finding a $\phi$ that satisfies these
constraints is even harder. However, if we truncate the game at depth $d$,
it becomes more feasible through linear programming. Notice that
(\ref{eq:immediate_potential}) (which holds for $d=\infty$) does not
explicitly define the value of $\phi(a,b,c)$: the recursion is
unbounded.  By truncating the game at $d$ we limit the available
states to $a,b \le d$.  Specifically, Miner 1 \emph{has} to capitulate
when $b=d$. 

The exact value of $d$ has a small effect on the game. In
particular, Miner 1 has more options as $d$ increases, since he does
not have to capitulate early. Therefore, the minimum $w$ needed for
compliance is a slightly increasing function of $d$. The effect of
this is quite muted however; the probability of having two branches
with length greater than $d$, without Miner 1 winning or capitulating,
decreases exponentially in $d$. In this section we select $d=8$, as it
is large enough to show the trend, but not too large to render the LP
computationally infeasible.

Relaxing the two $\max$es by inequalities, we get the
following LP:
\begin{equation*}
	\begin{array}{l}
		\text{minimize} \quad g + \frac{1}{D}
		\displaystyle\sum\limits_{a=0}^{d} \displaystyle\sum\limits_{b=0}^{d} 
		\displaystyle\sum\limits_{c=0}^{1} \phi(a,b,c)\\
		\text{subject to}\\
		\quad \phi(b+1,b,c) \ge \phi(0,0,0) + b + 1 + c\cdot w - g\\
		\quad\quad \text{ for } b < d, c \le 1\\
		\quad \phi(a,b,c) \ge \phi(0,s,1)\\
		\quad\quad \text{ for } a \le b < d, s < b, c \le 1\\
		\quad \phi(a,b,c) \ge p\cdot\phi(a+1,b,c) + (1-p)(\phi(a,b+1,c) - g)\\
		\quad\quad \text{ for } a \le b < d, c \le 1
	\end{array},
\end{equation*}
where $\phi(a,b,c) \ge 0$ and $D \gg d^2$ is a normalizing factor to
keep the sum of states insignificant compared to $g$. The constraints
are straightforward enough, but the objective is a minimization, which
might appear odd at first. The second term of the sum is only there to
ensure that feasible solutions only contain potential functions $\phi$
that tightly satisfy (\ref{eq:immediate_potential}). Since many states
could be unreachable (for example, if $w$ is high enough so that Miner
1 plays $\frontier$, the only mining states would be $(0,0,0)$ and
$(0,0,1)$), their exact value does not affect $g$ and we need this
'extra' term to keep them in check, just to bias the LP towards
certain feasible solutions. We minimize $g$ to use the following
lemma:
\begin{lemma}\label{lemma:induction}
	For any $g, \phi$ pair satisfying (\ref{eq:immediate_potential}) we 
	have: $$g_k(a,b,c) \le \phi(a,b,c) + k\cdot g.$$
\end{lemma}
\begin{proof}
	We use induction on $k$. Clearly, for $k=0$ we have
	$g_0(a,b,c) = 0 \le \phi(a,b,c)$.  For fixed $k$ we do strong
	induction on $b$ and backwards induction on $a$.  Starting with
	$b=0$, we begin the induction on $a=b+1$:
	\begin{align*}
		g_k(1,0,c) 
		&= g_{k-1}(0,0,0) + 1 + c\cdot w \\
		&\le \phi(0,0,0) + (k-1)g + 1 + c\cdot w\\
		&= \phi(0,0,0) + 1 + c\cdot w - g + k\cdot g\\
		&= \phi(1,0,c) + k\cdot g
	\end{align*}
	and 
	\begin{align*}
		g_k(0,0,c) 
		&= p\cdot g_k(1,0,c) + (1-p)g_{k-1}(0,1,c) \\
		&\le p(\phi(1,0,c) + k\cdot g) + (1-p)(\phi(0,1,c) + (k-1)g)\\
		&= p\cdot\phi(1,0,c) + (1-p)(\phi(0,1,c) - g) +k\cdot g\\
		&= \phi(0,0,c) + k\cdot g.
	\end{align*}
	For $b>0$ the proof works the same, starting from $g_k(b+1,b,c)$ as the 
	base case for $a$. For $a < b + 1$:
	\begin{align*}
		g_k(&a,b,c)
		= \max 
		\begin{cases}
			\max_{s=0,\ldots,b-1}g_k(0,s,1)\\
			pg_k(a+1,b,c) + (1-p)g_{k-1}(a,b+1,c) 
		\end{cases}\\
		&\le \max 
		\begin{cases}
			\max_{s=0,\ldots,b-1}\phi(0,s,1) + k g\\
			p\cdot(\phi(a+1,b,c)+k g) \\ \quad + (1-p)(\phi(a,b+1,c) + (k-1) g)
		\end{cases}\\
		&= k\cdot g + \max
		\begin{cases}
			\max_{s=0,\ldots,b-1}\phi(0,s,1)\\
			p\cdot \phi(a+1,b,c) \\ \quad + (1-p)(\phi(a,b+1,c) - g)
		\end{cases}\\
		&= \phi(a,b,c) + k g,
	\end{align*}
	where strong induction was used for the capitulation case.
\end{proof}

Using this lemma, we can find the minimum value of $w$ that leads to
\frontier~ being a best response for different values of $p$ by using
binary search on $w$, checking that the $\phi$ produced from the LP
satisfies (\ref{eq:immediate_potential}) and that $g = p +
p(1-p)w$. By the definition of $\totalgain$:
\begin{equation*}
	\totalgain = \lim_{k\rightarrow \infty} \frac{g_k(a,b,c)}{k} \le
	\lim_{k\rightarrow \infty} \frac{\phi(a,b,c) + k\cdot g}{k} 
	= p + p(1-p)w,
\end{equation*}
For $d=8$, we obtain the graph in Figure~\ref{fig:minimum_w1}.
\begin{figure}
	\centering
	\begin{tikzpicture}
	\begin{axis}[
	xlabel=$p$,
	ylabel=Minimum $w$,
	grid=major,
	legend pos= north west]
	\addplot[color=blue,mark=o] coordinates {
		(0.42, 0/1000)
		(0.43, 45/1000)
		(0.44, 132/1000)
		(0.45, 225/1000)
		(0.46, 329/1000)
		(0.47, 433/1000)
		(0.48, 549/1000)
		(0.49, 673/1000)
		(0.5, 797/1000)
	};
	\addlegendentry{$w$}
	
	\addplot[color=red,mark=o] coordinates {
		(0.42, 0/1000)
		(0.43, 28/1000)
		(0.44, 85/1000)
		(0.45, 146/1000)
		(0.46, 212/1000)
		(0.47, 283/1000)
		(0.48, 361/1000)
		(0.49, 444/1000)
		(0.5, 534/1000)
	};
	\addlegendentry{Shared $w'$}
	\end{axis}
	\end{tikzpicture}
	\caption{Minimum $w$ for the immediate release case.}
	\label{fig:minimum_w1}
\end{figure}
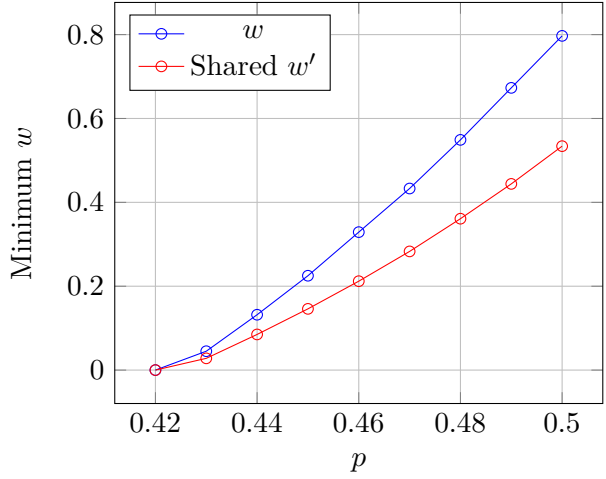

It is not a mistake that $w < 1$ for $p = 0.5$. The reason is that $d$
is finite. For $d = \infty$, the strategic miner with $p \ge 0.5$
never has a reason to capitulate, no matter $w$: the probability that
his branch eventually overtakes the honest one is always quite
high. However, the probability that he will win the race \emph{within
	$d$ steps} is much smaller, leading to this result.

If not all honest miners are required to pay forward exactly $w$, they
can share the costs. In particular, one strategy which also makes
Miner 1 comply with the protocol is for the first honest miner after
Miner 1 to pay forward $w'$ and every other honest miner to pay
forward $2w'$, until Miner 1 adds a block and claims the reward for himself, 
where the process
restarts. As we can see from the Figure~\ref{fig:minimum_w1}, although the total 
amount paid
forward is larger (as $2w' \ge w$), the maximum contribution of any
honest miner is $w' \le w$. We can verify compliance is indeed a best response 
for Miner 1 by using a slightly modified LP.

\subsection{Pure Nash Equilibria for small miners}
So far, we have only considered the best responses of Miner 1, given a
society of honest, identical miners and identified the minimum value
$w$ needed to ensure his compliance. The natural question is to
consider if this also a best response from the other miners. Clearly,
when treated as a group this is not their best response. If they do
not pay forward any amount, Miner 1 would have an average gain of
$g' \ge p$, by mining strategically. By paying forward enough, they
need to \emph{increase} his gain to $g^\star \ge g'$, as he can always
employ strategic mining while they also pay forward. Therefore, it is more costly 
for Miner 2 to pay forward $w$ and ensure compliance than to pay forward 
nothing and allow strategic mining.

We will argue however that under certain conditions paying forward can
be a pure Nash equilibrium, if Miner 2 consists of many \emph{small}
miners that are treated individually.  A miner is considered
\emph{small} if his chance of mining a block is negligible. Of course
at every step some small miner will mine a block with probability
$1-p$.  To see why paying forward is an equilibrium, consider the
options of a small miner at the \emph{exact} moment when he mines a
block: he can either pay forward and guarantee it is included in the
chain or keep everything for himself, but risk getting undercut by
Miner 1. Being a small miner, his payoff from mining another block in the future 
is insignificant compared to the current reward at stake. Formally, the following 
strategy, called
$\strictfrontier(w)$, is part of a PNE:
\begin{definition}[$\strictfrontier$]
	Every miner (except Miner 1) pays forward $w$ and mines at the end
	of the longest chain that contains blocks with pay forward values
	either $w$ or $0$
\end{definition}
Miner 1's strategy is to mine at the frontier without paying
forward, as before. Note, that paying forward some value other than
$w$ or $0$ \emph{cannot} be a best response, since that block will be
ignored by the honest miners, who have the majority of the
computational power\footnote{The reason $\strictfrontier$ accepts pay forward 
of $0$ is just to incentivize Miner 1 to cooperate. Without it, every block would 
have the same pay forward, negating any effects this may have and reducing to 
the classic bitcoin protocol.}

\begin{theorem}
	For $d=8$ and any $p \le 0.44$, there exists $w$ such that Miner 1
	playing $\frontier$ and every other small miner playing
	$\strictfrontier(w)$ is a pure Nash equilibrium.
\end{theorem}
\begin{proof}[Proof sketch]
	From the perspective of the small miner, he has two opportunities
	\emph{not} to pay forward: after another small miner or after Miner
	1. We can modify the LP from Section~\ref{sec:LP} to contain both
	these deviations by the small miner. Examining the tightness of the
	constraints, we can identify the mining states of Miner
	1. Specifically, for $p \le 0.44$, we observe that in both cases,
	after the small miner adds his block without paying forward, Miner 1
	immediately forks the chain behind the small miner and starts mining
	in parallel. In particular, states $(0,1,c),(1,1,c),(1,2,c)$ and
	$(2,2,c)$ become mining states and $(0,3,c),(1,3,c),(2,3,c)$
	capitulation states, for $c \in \sset{0,1}$. Essentially he mines
	while one step behind and capitulates if the honest miners add 2
	more blocks. Overall, he overtakes the honest branch with
	probability at least $p^2 + (1-p)p^3$.
	
	Therefore, the expected payoff of the small miner is at most
	$(1+w)(1-p^2 - (1-p)p^3)$, compared to at least $1 - w$ which would have 
	been
	his payoff had he paid forward $w$ (since in this case he would always keep 
	his mined block). Plugging in the $p,w$ pairs
	found in Figure~\ref{fig:minimum_w1} we verify that
	$\strictfrontier(w)$ is indeed a best response.
\end{proof}
It would be very interesting to see if there exists a strategy profile that is a PNE 
where Miner 1 plays $\frontier$ for $0.45 < p < 0.5$. The analysis seems quite 
challenging and would likely require different techniques to analyze the emerging 
stochastic game. Ideally, every small miner comprising Miner 1 would use the 
same $w$, but it could be the case that better bounds could be obtained if the 
pay forward amount use depends on the pay forward received from the previous 
block.

\section{Strategic Release}\label{sec:strategic_release}
\newcommand{\strategicgain}{\hat{g}^{\star}}
\newcommand{\honestgain}{\hat{g}^{\star}} Similarly to the immediate
release case, we identify the maximum hash power $p$ such that if all
miners but one follow $\frontier(w)$, the remaining miner's best
response is $\frontier$ if his hash power is less than $p$. As before,
it is enough to consider a two player game where Miner 1 is the
deviant miner with hash power $p < 0.5$ and Miner 2 represents all
small miners, as they follow the same strategy.

The blockchain retains its structure: it still contains a long trunk
of blocks followed by two branches, corresponding to the released
blocks of Miner 1 and Miner 2. Since Miner 1 could also have more
unreleased blocks in his branch the state is now a quadruple
$(a_r, a, b,c)$ where $a$ is the number of blocks in Miner 1's branch,
of which $a_r$ have been released, and $b$ is the number of blocks in
Miner 2's branch. Also, contrary to the immediate release case we can
have $a > b + 1$.

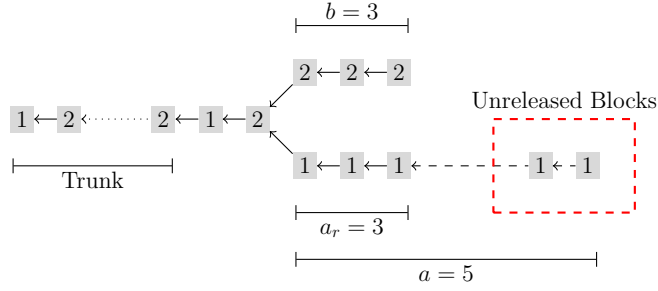
\begin{figure}
	\begin{center}
		\begin{tikzpicture}[scale=.62, transform shape]
		
		\tikzstyle{every node} = [rectangle, fill=gray!30,font=\Large]
		\node (t1) at (-3, 0) {1};
		\node (t2) at (-2, 0) {2};
		\node (t3) at (0, 0) {2};
		\node (t4) at (1, 0) {1};
		\node (t5) at (2, 0) {2};
		
		\node (b1) at (3, 1) {2};
		\node (b2) at (4, 1) {2};    
		\node (b3) at (5, 1) {2};    
		
		\node (ar1) at (3, -1) {1};        
		\node (ar2) at (4, -1) {1};    
		\node (ar3) at (5, -1) {1};
		\node (a4) at (8, -1) {1}; 
		\node (a5) at (9, -1) {1}; 
		
		\draw[black,<-] (t1) -- (t2);
		\draw[black,<-,dotted] (t2) -- (t3);
		\draw[black,<-] (t3) -- (t4);
		\draw[black,<-] (t4) -- (t5);
		
		\draw[black,<-] (t5) -- (b1);
		\draw[black,<-] (b1) -- (b2);
		\draw[black,<-] (b2) -- (b3);
		
		\draw[black,<-] (t5) -- (ar1);
		\draw[black,<-] (ar1) -- (ar2);
		\draw[black,<-] (ar2) -- (ar3);
		\draw[black,<-,dashed] (ar3) -- (a4);
		\draw[black,<-,dashed] (a4) -- (a5);
		
		\draw[dashed,thick,red] (7,0) rectangle (10,-2);
		\node [above,fill=white] at ($1/2*(7,.1) + 1/2*(10,.1)$) {Unreleased Blocks};
		
		\draw[black,|-|] (-3.2,-1) --node[draw=none,fill=none,below]{Trunk} (0.2,-1);
		
		\draw[black,|-|] (2.8,-2) --node[draw=none,fill=none,below]{$a_r=3$} (5.2,-2);
		
		\draw[black,|-|] (2.8,-3) --node[draw=none,fill=none,below]{$a=5$} (9.2,-3);
		
		\draw[black,|-|] (2.8,2) --node[draw=none,fill=none,above]{$b=3$} (5.2,2);
		\end{tikzpicture}
	\end{center}
	\caption[Tree representing state $(3,5,3,1)$]{This tree represents
		state $(3,5,3,1)$. Miner 1 has already mined 5 blocks ahead of the
		fork, but has only released 3 of them. Miner 2 knows this because
		the game is complete information. However, he cannot mine there
		until Miner 1 releases his blocks. Miner 1 can strategically
		release blocks right when Miner 2's branch is about to become the
		main one, thus wasting Miner 1's computational power.}
\end{figure}

Since Miner 2 follows the $\frontier(w)$ strategy, he will capitulate
if $a_r \ge b + 1$ and continue mining his branch for $a_r < b +
1$. Therefore, without loss of generality we can assume that at state
$(a_r,a,b,c)$ if $a < b + 1$ then $a_r = a$ and if $a \ge b + 1$ then
$a_r = b$, otherwise Miner 2 would immediately capitulate and the game
would continue at state $(0, a - a_r, 0,0)$. Therefore, we can encode
the states of the mining game with strategic release by the triplet
$(a,b,c)$, where $a_r = \min(a,b)$.

As before, we need to recursively define the gain of Miner 1 after the
longest branch is extended by $k$ levels. When $a < b + 1$, the
actions of Miner 1 are exactly the same as in the immediate release
case: he can either capitulate to $(0,s,1)$ or mine. However, for
$a \ge b + 1$ he can release one more block to force Miner 2 to
capitulate and continue at state $(a-b-1,0,0)$.
\begin{equation}\label{eq:gain_at_strategic}
	\hat{g}_k(a,b,c) = 
	\max
	\begin{cases}
		k + b + c\cdot w \quad\quad \text{ if } a \ge k + b\\
		\max
		\begin{cases}
			\max_{s=0,\ldots,b-1}\hat{g}_k(0,s,1)\\
			p\hat{g}_k(a+1,b,c) \\ \quad + (1-p)\hat{g}_{k-1}(a,b+1,c)\\ 
			\hat{g}_{k-1}(a-b-1,0,0) \\ \quad + b + 1 + c\cdot w
		\end{cases}
	\end{cases}
\end{equation}
where the last term applies for $a \ge b+1$. Equivalently, we can
define $g_k(a,0,c) = -\infty$ for all $k$ and $a < 0$. The first term
is necessary, as without it the recurrence is ill-defined and Miner 1
can keep mining forever. Once he reached the horizon $k$ he can safely
release all his blocks. We define the expected gain per level
$\strategicgain$ and potential $\hat{\phi}$ as we did in
(\ref{def:expected_gain}) for the immediate release.  Note that the
case where $a=k+b$ does not appear, since $\phi$ is the asymptotic
advantage as $k\rightarrow\infty$ and $a,b$ are bounded.
\begin{equation}\label{eqn:strategic_potential}
	\hat{\phi}(a,b,c) = 
	\max
	\begin{cases}
		&\max_{s=0,\ldots,b-1}\hat{\phi}(0,s,1)\\
		&p\hat{\phi}(a+1,b,c) + (1-p)(\hat{\phi}(a,b+1,c)-\strategicgain)\\ 
		&\hat{\phi}(a-b-1,0,0) + b + 1 + c\cdot w - \strategicgain
	\end{cases}.
\end{equation}
Releasing blocks only causes Miner 2 to capitulate if $a \ge b + 1$,
so we set $\phi(a,0,c) = -\infty$ for $a < 0$. Before moving on, we
need a useful inequality of the \emph{immediate release} potential
$\phi$, to use for bounding the advantage of some states relative to
others.
\begin{lemma}\label{lemma:potential_upper_bound}
	For nonnegative integers $a,b$ and $\ell \in \sset{0,1}$ we have that:
	\begin{equation}
		\phi(a+\ell,b+\ell,1) \le \phi(a,b,0) + 
		(\ell+w)\cdot\left(\frac{p}{1-p}\right)^{b-a+1}.
	\end{equation}
\end{lemma}
\begin{proof}
	We first need to establish that for $c \in \sset{0,1}$
	\begin{equation}\label{ineq:upper_bound_generic}
		g_k(a+\ell,b+\ell,c) \le g_k(a,b,c) +
		\ell\cdot\left(\frac{p}{1-p}\right)^{b-a+1} 
	\end{equation}
	and
	\begin{equation}\label{ineq:upper_bound_pf}
		g_k(a+\ell,b+\ell,1) \le g_k(a,b,0) +
		(\ell+w)\cdot\left(\frac{p}{1-p}\right)^{b-a+1}. 
	\end{equation}
	Suppose that from state $(a,b)$ Miner 1 follows the same strategy as
	state $(a+\ell,b+\ell)$. This is possible, as winning or
	capitulating depends only on the difference $b-a$. Let
	$\bar{g_k}(a,b,c)$ be the gain from playing this suboptimal strategy
	and $\bar{r}(a,b)$ the probability of winning from state
	$(a,b)$. Clearly, we have:
	\begin{equation}
		g_k(a,b,c) \ge \bar{g}_k(a,b,c) =  g_k(a+\ell,b+\ell,c) -
		\ell\cdot\bar{r}(a,b). 
	\end{equation}
	Given that $\bar{r}(a,b) \le (p/(1-p))^{b-a+1}$ from
	\cite[Lemma~1]{blockchain2016} we get the first inequality.
	
	For the second, let $r(a,b)$ be the probability of winning from
	$(a,b)$ with the optimal strategy. As before:
	\begin{equation}
		g_k(a,b,1) = g_k(a,b,0) + w\cdot r(a,b) \le g_k(a,b,0) + w\cdot 
		\left(\frac{p}{1-p}\right)^{b-a+1},
	\end{equation}
	which we substitute in \eqref{ineq:upper_bound_generic}.
	
	To complete the proof, we take limits in
	(\ref{ineq:upper_bound_generic}) and use the definition of $\phi$.
\end{proof}

In \cite{blockchain2016} it was shown that for $p \le 0.361$
$\frontier$ is a best response, therefore (for $w=0$) there exists a
$\phi$ such that for $\totalgain = p$ (\ref{eq:immediate_potential})
is satisfied, meaning that honest mining maximizes the fraction of
blocks added by Miner 1. By Lemma~\ref{thm:frontier_more_payforward},
increasing $w$ can only strengthen this result, therefore for
$p \le 0.361$ and any $w \le 1$ there exists $\phi$ that satisfies
(\ref{eq:immediate_potential}) for $\totalgain = p + p(1-p)w$. Using
this, we extend the potential to states $(a,b,c)$ with $a > b + 1$:
\begin{equation}\label{eqn:strategic_potential_adjusted}
	\bar{\phi}(a,b,c) = 
	\begin{cases}
		\phi(a,b,c) & \text{if } a \le b + 1\\
		a \lambda + b \mu + \kappa + c \cdot w & \text{ otherwise}\\
	\end{cases},
\end{equation}
where $\lambda = \frac{(p-1)^2 (1-p w)}{1-2 p}$,
$\mu = \frac{p \left(p-(p-1)^2 w\right)}{2 p-1}$ and
$\kappa = \frac{(p-1) p (p w-1)}{2 p-1}$. The constants have been
selected so that
$\bar{\phi}(a,b,c) = p \bar{\phi}(a+1,b,c) + (1-p)(\bar{\phi}(a,b+1,c)
- \honestgain)$ and
$\bar{\phi}(b+1,b,c) = b + 1 + c\cdot w - \honestgain =
\phi(b+1,b,c)$, for $\honestgain = p + p(1-p)w$ which is the honest
gain. Doing this, we can use known results about $\phi$ for smaller
states while having a tight enough closed form of $\hat{\phi}$ on
states which are unlikely to be reached.

Following the notation of \cite{blockchain2016} we define 
$\bar{\phi}_M$ for when Miner 1 continues to mine, $\bar{\phi}_R$ when Miner 1 
releases some blocks and $\bar{\phi}_C$ when he capitulates.
\begin{align*}
	\bar{\phi}_M(a,b,c)
	&= p\cdot \bar{\phi}(a+1,b,c) + (1-p)(\bar{\phi}(a,b+1,c) - 
	\honestgain)\\
	\bar{\phi}_R(a,b,c)
	&= \bar{\phi}(a-b-1,0,0) + b + 1 + c\cdot w - \honestgain\\
	\bar{\phi}_C(a,b,c) &= \max_{s=0,\ldots,b-1} \bar{\phi}(0,s,1).
\end{align*}

\begin{theorem}
	For every $p < 0.344$, there exists $w \ge 0$ large enough so that
	if every miner but one follows $\frontier(w)$, the best response of
	the remaining miner with hash power $p$ is $\frontier$.
\end{theorem}
\begin{proof}
	We need to show that $\bar{\phi}$ is a valid potential and satisfies 
	(\ref{eqn:strategic_potential}) for $\strategicgain = p + p(1-p)w$, or 
	equivalently:
	\begin{equation*}
		\bar{\phi}(a,b,c) = \max\{\bar{\phi}_M(a,b,c), \bar{\phi}_R(a,b,c), 
		\bar{\phi}_C(a,b,c)\}
	\end{equation*}

	\begin{lemma}\label{lemma:satisfies_req}
		The potential $\bar{\phi}$ and $\strategicgain = p + p(1-p)w$
		satisfy recurrence (\ref{eqn:strategic_potential}) for
		$p \le 0.344$ (root of the polynomial
		$- p^4 + 3p^3 - 7 p^2 + 5p-1$) and some $w < 1$.
	\end{lemma}
	\begin{proof}
		\begin{claim}
			For states $(a,b,c)$ with $a < b + 1$:
			\begin{equation*}
				\bar{\phi}(a,b,c) = \phi(a,b,c) = \max\{\bar{\phi}_M(a,b,c), 
				\bar{\phi}_R(a,b,c), \bar{\phi}_C(a,b,c)\}.
			\end{equation*}
		\end{claim}
		In this case $\bar{\phi}_R(a,b,c) = -\infty$ and therefore
		$\bar{\phi}(a,b,c) = \phi(a,b,c)$ which satisfies
		(\ref{eq:immediate_potential}) and (\ref{eqn:strategic_potential})
		by definition, as releasing is only possible for $a \ge b + 1$.
		\begin{claim}
			For states $(a,b,c)$ with $a > b + 1$:
			\begin{equation*}
				\bar{\phi}(a,b,c) = \bar{\phi}_M(a,b,c) = \max\{\bar{\phi}_M(a,b,c), 
				\bar{\phi}_R(a,b,c), \bar{\phi}_C(a,b,c)\}.
			\end{equation*}
		\end{claim}
		By definition, $\bar{\phi}(a,b,c) = \bar{\phi}_M(a,b,c)$. We have:
		\begin{equation*}
			\bar{\phi}(a,b,c) - \bar{\phi}_R(a,b,c) = 
			\frac{b (2-4p)+(1-p) (2-3 p) (1-p w)}{1-2 p} > 0. 
		\end{equation*}
		Since $p \le 0.344$ we know that $\phi(a,b,c)$ corresponds to the
		potential of the honest mining strategy, in the immediate release
		case. Therefore:
		\begin{align*}
			\bar{\phi}_C(a,b,c) &= \max_{s=0,\ldots,b-1} \bar{\phi}(0,s,1)\\
			&= \max_{s=0,\ldots,b-1} \phi(0,s,1) = \phi(0,0,1) \le w \cdot
			\frac{p}{1-p} 
		\end{align*}
		as $(0,0,1)$ and $(0,0,0)$ are the only mining states and by using
		Lemma~\ref{lemma:potential_upper_bound} for the last
		inequality. Also:
		\begin{align*}
			\bar{\phi}_R(a,b,c)
			&= \bar{\phi}(a-b-1,0,0) + b + 1 + c\cdot w - \honestgain \\
			&\ge 1 - p(1-p)w \\
			&\ge w \cdot \frac{p}{1-p} \\
			&\ge \bar{\phi}_C(a,b,c), 
		\end{align*}
		for $p < 0.344$.
		
		\begin{claim}
			For states $(b+1,b,c)$:
			\begin{align*}
				\bar{\phi}(b+1,&b,c) = \bar{\phi}_R(b+1,b,c)\\
				&= \max\{\bar{\phi}_M(b+1,b,c),  
				\bar{\phi}_M(b+1,b,c), \bar{\phi}_C(b+1,b,c)\}.
			\end{align*}
		\end{claim}
		Exactly as before, we have
		$\bar{\phi}_R(b+1,b,c) = b + 1 + c\cdot w - \honestgain \ge
		\bar{\phi}_C(b+1,b,c)$.  For $\bar{\phi}_M(b+1,b,c)$:
		\begin{align*}
			\bar{\phi}_M(b+1,b,c) 
			&= 
			p \bar{\phi}(b+2,b,c) +(1-p) \left(\bar{\phi}(b+1,b+1,c)-g\right)\\
			&\le
			p ((b+2) \lambda +b \mu +\kappa + c\cdot w)\\
			&\quad \quad+(1-p)\left((b+1+c\cdot w)\frac{p}{1-p}-g\right),
		\end{align*}
		by the definition of $\bar{\phi}$ and
		Lemma~\ref{lemma:potential_upper_bound}.  Then:
		\begin{align*}
			\bar{\phi}_R(&a,b,c)
			- \bar{\phi}_M(a,b,c) \ge \\
			&b (1-2 p)^2+(p-1) \left(p^3 
			w-p^2 
			(w+1)+p 
			(2 c\cdot w+5)-c\cdot w\right)\\
			&\quad+1.
		\end{align*}
		Since $b \ge 0$ we only need:
		$$(p-1) \left(p^3 w-p^2 (w+1)+p (2 c\cdot w+5)-c\cdot w\right)+1
		\ge 0.$$ Setting $w=1$, it holds for $p \le 0.344$ for $c = 0$ and
		$p \le 0.442$ for $c=1$.  For smaller values of $p$ it is not
		necessary to set $w$ to it's highest value.
	\end{proof}
	Now, we can use the equivalent of Lemma~\ref{lemma:induction}(whose
	proof has very minor differences) for the strategic release case to
	get that:
	\begin{align*}
		\strategicgain = \lim_{k\rightarrow \infty} \frac{\hat{g}_k(a,b,c)}{k} &\le
		\lim_{k\rightarrow \infty} \frac{\bar{\phi}(a,b,c) + k\cdot(p + p(1-p)w)}{k} \\
		&= p + p(1-p)w.
	\end{align*}
	This shows that his gain is at most what he would get by playing
	$\frontier$, hence a best response.
\end{proof}

Through a procedure similar to Section~\ref{sec:LP}, adjusting the LP
for this case, we obtain the following graph for $d = 8$ for the
minimum value of $w$ required.

\begin{figure}[H]
	\begin{center}
		\begin{tikzpicture}
		\begin{axis}[
		xlabel=$p$,
		ylabel=Minimum $w$,
		grid=major]
		\addplot[color=blue,mark=o] coordinates {
			(0.33, 0)
			(0.34, 134/1000)
			(0.35, 338/1000)
			(0.36, 540/1000)
			(0.37, 721/1000)
			(0.38, 879/1000)
			(0.385, 1)
		};
		\end{axis}
		\end{tikzpicture}
	\end{center}
	\caption{Minimum values of $w$ for the strategic release case.}
\end{figure}
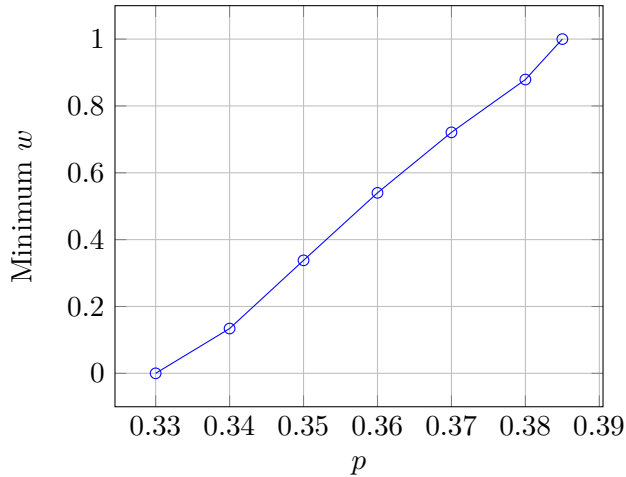
Contrary to Figure~\ref{fig:minimum_w1}, for $p=0.385$ we have
$w \approx 1$. As $d \rightarrow \infty$, we reach $w = 1$ for
$p \approx 0.38$.

\bibliographystyle{plain}
\bibliography{bibliography}

\end{document}